\documentclass[11pt,fleqn]{article}
\usepackage[utf8]{inputenc}
\usepackage[top=1in,bottom=1in,left=1in,right=1in]{geometry}
\usepackage{fouriernc}
\usepackage{authblk}
\usepackage{mathtools}
\usepackage[all,cmtip]{xy}
\usepackage{multicol}
\usepackage{amssymb}
\usepackage{amsthm}
\usepackage[all]{xy}
\newcommand{\ket}[1]{\rvert#1\rangle}
\newcommand{\bra}[1]{\langle #1\rvert}

\newcommand{\Braket}[3]{\bra{#1}#2\ket{#3}}
\newtheorem{Proposition}{Proposition}
\title{A Dirac--Dunkl equation on $S^2$ \\ and the Bannai--Ito algebra}
\author[1]{Hendrik De Bie}
\author[2]{Vincent X. Genest}
\author[2]{Luc Vinet}
\affil[1]{Department of Mathematical Analysis, Faculty of Engineering and Architecture, Ghent University, Galglaan 2, 9000 Gent, Belgium}
\affil[2]{Centre de recherches math\'ematiques, Universit\'e de Montr\'eal, P.O. Box 6128, Centre-ville Station, Montr\'eal, Canada H3C 3J7}

\date{}

\begin{document}
\maketitle
\thispagestyle{empty}
\hrule
\begin{abstract}
\noindent
The Dirac--Dunkl operator on the 2-sphere associated to the $\mathbb{Z}_2^3$ reflection group is considered. Its symmetries are found and are shown to generate the Bannai--Ito algebra. Representations of the Bannai--Ito algebra are constructed using ladder operators. Eigenfunctions of the spherical Dirac-Dunkl operator are obtained using a Cauchy--Kovalevskaia extension theorem. These eigenfunctions, which correspond to Dunkl monogenics, are seen to support finite-dimensional irreducible representations of the Bannai--Ito algebra.
\smallskip

\smallskip

\noindent \textbf{AMS classification numbers:} 81Q05, 81R99
\end{abstract}
\hrule
\section{Introduction}
The purpose of this paper is to study the Dirac--Dunkl operator on the two-sphere for the $\mathbb{Z}_2^3$ reflection group and to investigate its relation with the Bannai--Ito algebra.

The Bannai--Ito algebra is the associative algebra over the field of real numbers with generators $I_1$, $I_2$, and $I_3$ satisfying the relations
\begin{align}
\label{BI-Algebra}
 \{I_1,I_2\}=I_3+\alpha_3,\quad \{I_2,I_3\}=I_1+\alpha_1,\quad \{I_3,I_1\}=I_2+\alpha_2,
\end{align}
where $\{a,b\}=ab+ba$ is the anticommutator and where $\alpha_i$, $i=1,2,3$, are real structure constants. The algebra \eqref{BI-Algebra} was first presented in \cite{2012_Tsujimoto&Vinet&Zhedanov_AdvMath_229_2123} as the algebraic structure encoding the bispectral properties of the Bannai--Ito polynomials, which together with the Complementary Bannai--Ito polynomials are the parents of the family of $-1$ polynomials \cite{2013_Genest&Vinet&Zhedanov_SIGMA_9_18,2012_Tsujimoto&Vinet&Zhedanov_AdvMath_229_2123}. The Bannai--Ito algebra also arises in representation theoretic problems \cite{2014_Genest&Vinet&Zhedanov_ProcAmMathSoc_142_1545} and in superintegrable systems \cite{2014_Genest&Vinet&Zhedanov_JPhysA_47_205202}; see \cite{2015_DeBie&Genest&Tsujimoto&Vinet&Zhedanov} for a recent overview.

Following their introduction in \cite{1988_Dunkl_MathZeit_197_33, 1989_Dunkl_TransAmerMathSoc_311_167, 1991_Dunkl_CanJMath_43_1213}, Dunkl operators have appeared in various areas. They enter the study of Calogero--Moser--Sutherland models \cite{2000_VanDiejen&Vinet}, they play a central role in the theory of multivariate orthogonal polynomials associated to reflection groups \cite{2014_Dunkl&Xu}, they give rise to families of stochastic processes \cite{2008_Graczyk&Rosler&Yor, 1998_Rosler&Voit_AdvApplMath_21_575}, and they can be used to construct quantum superintegrable systems involving reflections \cite{2013_Genest&Ismail&Vinet&Zhedanov_JPhysA_46_145201,2014_Genest&Ismail&Vinet&Zhedanov_CommMathPhys_329_999}. Dunkl operators also find applications in harmonic analysis and integral transforms \cite{2013_Dai&Xu,2003_Rosler}, as they naturally lead to the Laplace-Dunkl operators, which are second-order differential/difference operator that generalize the standard Laplace operator.

In a recent paper \cite{2015_Genest&Vinet&Zhedanov_CommMathPhys}, the analysis of the Laplace-Dunkl operator on the two-sphere associated to the $\mathbb{Z}_2\times \mathbb{Z}_2\times \mathbb{Z}_2$ Abelian reflection group was cast in the frame of the Racah problem for the Hopf algebra $sl_{-1}(2)$ \cite{2011_Tsujimoto&Vinet&Zhedanov_SIGMA_7_93}, which is closely related to the Lie superalgebra $\mathfrak{osp}(1|2)$. It was established that the Laplace-Dunkl operator on the two-sphere $\Delta_{S^2}$ can be expressed as a quadratic polynomial in the Casimir operator corresponding to the three-fold tensor product of unitary irreducible representations of $sl_{-1}(2)$. A central extension of the Bannai--Ito algebra was seen to emerge as the invariance algebra for $\Delta_{S^2}$ and subspaces of the space of Dunkl harmonics that transform according to irreducible representations of the Bannai--Ito algebra were identified.

It is well known that the square root of the standard Laplace operator is the Dirac operator, which is a Clifford-valued first order differential operator. The study of Dirac operators is at the core of Clifford analysis, which can be viewed as a refinement of harmonic analysis \cite{1982_Brackx&Delanghe&Sommen}. Dirac operators also lend themselves to generalizations involving Dunkl operators \cite{2012_DeBie&Orsted&Somberg&Soucek_TransAmerMathSoc_364_3875, 2006_Cerejeiras&Kahler&Ren_ComplexVar&EllEqs_51_487, 2009_Orsted&Somberg&Soucek_AdvApplCliffAlg_19_403}. These so-called Dirac-Dunkl operators are the square roots of the corresponding Laplace-Dunkl operators and as such, they exhibit additional structure which makes their analysis both interesting and enlightening. 

In this paper, the Dirac-Dunkl operator on the two-sphere associated to the $\mathbb{Z}_2^3$ Abelian reflection group will be examined. We shall begin by discussing the Laplace-- and Dirac-- Dunkl operators in $\mathbb{R}^3$. The Dirac--Dunkl operator will be defined in terms of the Pauli matrices, which play the role of Dirac's gamma matrices for the three-dimensional Euclidean space. It will be shown that the Laplace-- and Dirac-- Dunkl operators can be embedded in a realization of $\mathfrak{osp}(1|2)$. The notion of Dunkl monogenics, which are homogeneous polynomial null solutions of the Dirac--Dunkl operator, will be reviewed as well as the corresponding Fischer theorem, which describes the decomposition of the space of homogeneous polynomials in terms of Dunkl monogenics. The Dirac--Dunkl operator on the two-sphere, to be called spherical, will then be defined in terms of generalized ``angular momentum'' operators written in terms of Dunkl operators and its relation with the spherical Laplace--Dunkl operator will be made explicit. The algebraic interpretation of the Dirac--Dunkl operator will proceed from noting its connection with the sCasimir operator of $\mathfrak{osp}(1|2)$. The symmetries of the spherical Dirac--Dunkl operator will be determined. Remarkably, these symmetries will be seen to satisfy the defining relations of the Bannai--Ito algebra. The relevant finite-dimensional unitary irreducible representations of the Bannai--Ito algebra will be constructed using ladder operators. An explicit basis for the eigenfunctions of the spherical Dirac--Dunkl operator will be obtained. The basis functions, which span the space of Dunkl monogenics, will be constructed systematically using a Cauchy--Kovalevskaia (CK) extension theorem. It will be shown that these spherical wavefunctions, which generalize spherical spinors, transform irreducibly under the action of the Bannai--Ito algebra.

The paper is divided as follows.
\begin{itemize}
 \item Section 2: Dirac-- and Laplace-- Dunkl operators in $\mathbb{R}^3$ and $S^2$, $\mathfrak{osp}(1|2)$ algebra
 \item Section 3: Symmetries of the Dirac--Dunkl operator on $S^2$, Bannai--Ito algebra
 \item Section 4: Ladder operators, Representations of the Bannai--Ito algebra
 \item Section 5: CK extension, Eigenfunctions of the spherical Dirac--Dunkl operator
\end{itemize}
\section{Laplace-- and Dirac-- Dunkl operators for $\mathbb{Z}_2^3$}
In this section, we introduce the Laplace-- and Dirac-- Dunkl operators associated to the $\mathbb{Z}_2^3$ reflection group. We show that these operators can be embedded in a realization of $\mathfrak{osp}(1|2)$. We define the Dunkl monogenics and the Dunkl harmonics and review the Fischer decomposition theorem. We introduce the spherical Laplace-- and Dirac-- Dunkl operators and we give their relation.
\subsection{Laplace-- and Dirac-- Dunkl operators in $\mathbb{R}^3$}
Let $\vec{x}=(x_1,x_2,x_3)$ denote the coordinate vector in $\mathbb{R}^3$ and let $\mu_i$, $i=1,2,3$, be real numbers such that $\mu_i\geqslant 0$. The Dunkl operators associated to the $\mathbb{Z}_2^3$ reflection group, denoted by $T_i$, are given by
\begin{align}
\label{Dunkl-Operators}
  T_i=\partial_{x_i}+\frac{\mu_i}{x_i}(1-R_i),\qquad i=1,2,3,
\end{align}
where 
\begin{align*}
 R_i f(x_i)=f(-x_i)
\end{align*}
is the reflection operator. It is obvious that the operators $T_i$, $T_{j}$ commute with one another. We define the Dirac--Dunkl operator in $\mathbb{R}^3$, to be denoted by $\underline{D}$, as follows:
\begin{align}
\label{Dunkl-Dirac-R3}
 \underline{D}=\sigma_1 T_1+\sigma_2 T_2+\sigma_3 T_3,
\end{align}
where the $\sigma_i$ are the familiar Pauli matrices
\begin{align*}
 \sigma_1=
 \begin{pmatrix}
  0 & 1
  \\
  1 & 0
 \end{pmatrix}, \quad
\sigma_2=
\begin{pmatrix}
 0 & -i
 \\
 i & 0
\end{pmatrix},\quad
\sigma_3=
\begin{pmatrix}
 1 & 0
 \\
 0 & -1
\end{pmatrix}.
\end{align*}
These matrices satisfy the identities
\begin{align*}
 [\sigma_i,\sigma_j]=2i\epsilon_{ijk}\sigma_k,\qquad \sigma_i\sigma_j=i\epsilon_{ijk}\sigma_k+\delta_{ij},
\end{align*}
where $[a,b]=ab-ba$ is the commutator, where $\epsilon_{ijk}$ is the Levi-Civita symbol and where summation over repeated indices is implied. The Pauli matrices provide a representation of the Euclidean Clifford algebra with three generators on $\mathbb{C}^2$, i.e. on the space of two-spinors. Indeed, one has
\begin{align}
\label{Clifford}
 \{\sigma_i,\sigma_j\}=2\delta_{ij},\qquad i,j=1,2,3.
\end{align}
As a direct consequence of \eqref{Clifford}, one has
\begin{align}
\label{Laplace-Dunkl-R3}
 \underline{D}^2=\Delta=T_1^2+T_2^2+T_3^2,
\end{align}
where $\Delta$ is the Laplace--Dunkl operator in $\mathbb{R}^3$. 

The Dirac--Dunkl and the Laplace--Dunkl operators \eqref{Dunkl-Dirac-R3} and \eqref{Laplace-Dunkl-R3} can be embedded in a realization of the Lie superalgebra $\mathfrak{osp}(1|2)$. Let $\underline{x}$ and $||\vec{x}||^2$ be the operators defined by
\begin{align*}
\underline{x}=\sigma_1 x_1+\sigma_2 x_2+\sigma_3 x_3,\qquad ||\vec{x}||^2=\underline{x}^2=x_1^2+x_2^2+x_3^2,
\end{align*}
and let $\mathbb{E}$ stand for the Euler (or dilation) operator
\begin{align*}
 \mathbb{E}=x_1 \partial_{x_1}+x_2 \partial_{x_2}+x_3 \partial_{x_3}.
\end{align*}
A direct calculation shows that one has
\begin{alignat}{3}
\label{OSP12}
\begin{aligned}
&\{\underline{x},\underline{x}\}= 2\,||\vec{x}||^2,&\quad &\{\underline{D},\underline{D}\}=2\,\Delta,&\quad &\{\underline{x},\underline{D}\}=2\,(\mathbb{E}+\gamma_3),
\\
&[\underline{D},||\vec{x}||^2]=2\,\underline{x},&\quad &[\mathbb{E}+\gamma_3, \underline{x}]=\underline{x},&\quad &[\mathbb{E}+\gamma_3, \underline{D}]=-\underline{D},\qquad [\Delta,\underline{x}]=2\,\underline{D},
\\
&[\mathbb{E}+\gamma_3, \Delta]=-2\,\Delta,&\quad &[\mathbb{E}+\gamma_3, ||\underline{x}||^2]=2\,||\vec{x}||^2,&\quad &[\Delta, ||\vec{x}||^2]=4\,(\mathbb{E}+\gamma_3),
\end{aligned}
\end{alignat}
where
\begin{align*}
 \gamma_3=\mu_1+\mu_2+\mu_3+3/2.
\end{align*}
The commutation relations \eqref{OSP12} are seen to correspond to those of the $\mathfrak{osp}(1|2)$ Lie superalgebra \cite{2000_Frappat&Sciarrino&Sorba}. In fact, the relations \eqref{OSP12} hold in any dimension and for any choice of the reflection group with different values of the constant $\gamma$ \cite{2012_DeBie&Orsted&Somberg&Soucek_TransAmerMathSoc_364_3875, 2009_Orsted&Somberg&Soucek_AdvApplCliffAlg_19_403}.

Let $\mathcal{P}_{N}(\mathbb{R}^3)$ denote the space of homogeneous polynomials of degree $N$ in $\mathbb{R}^3$, where $N$ is a non-negative integer. The space of \emph{Dunkl monogenics} of degree $N$ for the reflection group $\mathbb{Z}_2^3$ shall be denoted by $\mathcal{M}_{N}(\mathbb{R}^3)$. It is defined as
\begin{align}
\label{Dunkl-Monogenics}
 \mathcal{M}_{N}(\mathbb{R}^3):=\mathrm{Ker}\,\underline{D}\;\bigcap\;(\mathcal{P}_{N}(\mathbb{R}^3)\otimes \mathbb{C}^2).
\end{align}
Similarly, the space of scalar \emph{Dunkl harmonics} of degree $N$ for the reflection group $\mathbb{Z}_2^3$ is denoted by $\mathcal{H}_{N}(\mathbb{R}^3)$ and defined as \cite{2014_Dunkl&Xu}
\begin{align*}
 \mathcal{H}_{N}(\mathbb{R}^3):=\mathrm{Ker}\,\Delta\;\bigcap\;\mathcal{P}_{N}(\mathbb{R}^3).
\end{align*}
The space of spinor-valued Dunkl harmonics has a direct sum decomposition in terms of the Dunkl monogenics. This decomposition reads
\begin{align*}
 \mathcal{H}_{N}(\mathbb{R}^3)\otimes \mathbb{C}^2=\mathcal{M}_{N}(\mathbb{R}^3)\;\oplus\;\underline{x}\,\mathcal{M}_{N-1}(\mathbb{R}^3).
\end{align*}
For $\gamma_3>0$, which is automatically satisfied when $\mu_1,\mu_2,\mu_3\geqslant 0$, the following direct sum decomposition holds \cite{2009_Orsted&Somberg&Soucek_AdvApplCliffAlg_19_403}:
\begin{align}
\label{Fischer-Decomposition}
 \mathcal{P}_{N}(\mathbb{R}^3)\otimes \mathbb{C}^2=\bigoplus_{k=0}^{N}\underline{x}^{k}\mathcal{M}_{N-k}(\mathbb{R}^3).
\end{align}
The above is called the \emph{Fischer decomposition} and will play an important role in what follows.
\subsection{Laplace-- and Dirac-- operators on $S^2$}
The explicit expression for the Dunkl operators \eqref{Dunkl-Operators} allows to write the Laplace--Dunkl operator \eqref{Laplace-Dunkl-R3} as
\begin{align}
\label{Laplace-Dunkl-2}
 \Delta=\sum_{i=1}^{3} \partial_{x_i}^2+\frac{2\mu_i}{x_i} \partial_{x_i}-\frac{\mu_i}{x_i^2}(1-R_i).
\end{align}
Since reflections are elements of $O(3)$, the Laplace--Dunkl operator \eqref{Laplace-Dunkl-2}, like the standard Laplace operator, separates in spherical coordinates. Consequently, it can be restricted to functions defined on the unit sphere. Let $\Delta_{S^2}$ denote the restriction of \eqref{Laplace-Dunkl-2} to the two-sphere, which shall be referred to as the spherical Laplace--Dunkl operator. It is seen that $\Delta_{S^2}$ can be expressed as
\begin{align}
\label{Laplace-Dunkl-S2}
\Delta_{S^2}=||\vec{x}||^2\Delta-\mathbb{E}\,(\mathbb{E}+2\mu_1+2\mu_2+2\mu_3+1).
\end{align}
The spherical Laplace--Dunkl operator can also be written in terms of the Dunkl angular momentum operators. These operators are defined as
\begin{align}
\label{Dunkl-J}
 L_1=\frac{1}{i}(x_2 T_3-x_3 T_2),\quad L_2=\frac{1}{i}(x_3 T_1-x_1 T_3),\quad L_3=\frac{1}{i}(x_1 T_2-x_2 T_1),
\end{align}
and satisfy the commutation relations 
\begin{align}
\label{Dunkl-J-Commutation}
 [L_i, L_j]=i\epsilon_{ijk}\;L_k(1+2\mu_k R_k),\qquad [L_i, R_i]=0,\qquad \{L_i, R_{j}\}=0.
\end{align}
Taking into account the relation \eqref{Laplace-Dunkl-S2}, a direct calculation shows that \cite{2014_Genest&Vinet&Zhedanov_JPhysConfSer_512_012010} 
\begin{align}
\label{Laplace-Dunkl-S2-J2}
 -\Delta_{S^2}=L_1^2+L_2^2+L_3^2-2\sum_{1\leqslant i<j\leqslant 3}\mu_i\mu_j\,(1-R_iR_j)-\sum_{1\leqslant j\leqslant 3}\mu_j\,(1-R_{j}).
\end{align}
When $\mu_1=\mu_2=\mu_3=0$, the relation \eqref{Laplace-Dunkl-S2-J2} reduces to the standard relation between the Laplace operator on the two-sphere and the angular momentum operators.

Let us now introduce the main object of study: the Dirac--Dunkl operator on the two-sphere. This operator, denoted by  $\Gamma$, is defined as
\begin{align}
\label{Dirac-Dunkl-S2}
 \Gamma=\vec{\sigma}\cdot \vec{L}+\vec{\mu}\cdot\vec{R},
\end{align}
where $\vec{\mu}=(\mu_1,\mu_2,\mu_3)$ and $\vec{R}=(R_1,R_2,R_3)$. When $\mu_1=\mu_2=\mu_3=0$, all reflections disappear and \eqref{Dirac-Dunkl-S2} reduces to the standard Hamiltonian describing spin-orbit interaction. The operator \eqref{Dirac-Dunkl-S2} is linked to the spherical Laplace--Dunkl operator by a quadratic relation. Upon using the equations \eqref{Clifford}, \eqref{Dunkl-J-Commutation} and \eqref{Laplace-Dunkl-S2-J2}, one finds that
\begin{align}
\label{Quad-Relation}
 \Gamma^2+\Gamma=-\Delta_{S^2}+(\mu_1+\mu_2+\mu_3)(\mu_1+\mu_2+\mu_3+1).
\end{align}
A relation akin to \eqref{Quad-Relation} was derived in \cite{2015_Genest&Vinet&Zhedanov_CommMathPhys}; it involved a scalar operator instead of $\Gamma$. The Dirac--Dunkl operator on the two-sphere has a natural algebraic interpretation in terms of the realization \eqref{OSP12} of the $\mathfrak{osp}(1|2)$ algebra. It corresponds, up to an additive constant, to the so-called sCasimir operator. Indeed, it is verified that 
\begin{align*}
 \{\Gamma+1, \underline{x}\}=0,\qquad \{\Gamma+1,\underline{D}\}=0,
\end{align*}
and that
\begin{align*}
 [\Gamma+1, \mathbb{E}]=0,\quad [\Gamma+1,||\vec{x}||^2]=0,\quad [\Gamma+1,\Delta]=0.
\end{align*}
Hence $\Gamma+1$ anticommutes with the odd generators and commutes with the even generators of $\mathfrak{osp}(1|2)$, which is the defining property of the sCasimir operator \cite{1995_Lesniewski_JMathPhys_36_1457}. The spherical Dirac--Dunkl operator is usually written as a commutator (see for example \cite{2011_DeBie&DeSchepper_BullBelMathSoc_18_193}). For \eqref{Dirac-Dunkl-S2}, one has
\begin{align}
\label{Gamma-Commutator}
 \Gamma+1=\frac{1}{2}\big([\underline{D},\underline{x}]-1\big).
\end{align}

The space of Dunkl monogenics $\mathcal{M}_{N}(\mathbb{R}^3)$ of degree $N$ is an eigenspace for this operator. Indeed, upon using \eqref{Gamma-Commutator}, the $\mathfrak{osp}(1|2)$ relations \eqref{OSP12} and the fact that
\begin{align*}
 \underline{D}\,\mathcal{M}_{N}(\mathbb{R}^3)=0,\qquad \mathbb{E}\,\mathcal{M}_{N}(\mathbb{R}^3)=N \mathcal{M}_{N}(\mathbb{R}^3),
\end{align*}
one can write
\begin{align*}
 (\Gamma+1)\,\mathcal{M}_{N}(\mathbb{R}^3)&=\frac{1}{2}\Big([\underline{D},\underline{x}]-1\Big)\mathcal{M}_{N}(\mathbb{R}^3)
 =\frac{1}{2}\Big(\underline{D}\;\underline{x}-1\Big)\mathcal{M}_{N}(\mathbb{R}^3)
 \\
 &=\frac{1}{2}\Big(\{\underline{x},\underline{D}\}-1\Big)\mathcal{M}_{N}(\mathbb{R}^3)
=\frac{1}{2}\Big(2(\mathbb{E}+\gamma_3)-1\Big)\mathcal{M}_{N}(\mathbb{R}^3),
\end{align*}
which gives
\begin{align}
\label{Gamma-Eigenvalues}
 (\Gamma+1)\;\mathcal{M}_{N}(\mathbb{R}^3)=(N+\mu_1+\mu_2+\mu_3+1)\;\mathcal{M}_{N}(\mathbb{R}^3),
\end{align}
where $N=0,1,2,\ldots$ is a non-negative integer.
\section{Symmetries of the spherical Dirac--Dunkl operator}
In this section, the symmetries of the spherical Dirac--Dunkl operator are obtained and are seen to satisfy the defining relations of the Bannai--Ito algebra.

Introduce the operators $J_i$ defined by
\begin{align}
\label{Symmetries}
 J_i&=L_i+\sigma_i(\mu_j R_j+\mu_k R_k+1/2),\qquad i=1,2,3,
\end{align}
where $(ijk)$ is a cyclic permutation of $\{1,2,3\}$. The operators $J_i$ are symmetries of the spherical Dirac--Dunkl operator, as it is verified that
\begin{align*}
 [\Gamma,J_i]=0,\qquad i=1,2,3.
\end{align*}
The operator $\Gamma$ can be expressed in terms of the symmetries $J_i$ in the following way:
\begin{align*}
 \Gamma=\sigma_1 J_1+\sigma_2 J_2+\sigma_3 J_3-\mu_1 R_1-\mu_2 R_2-\mu_3 R_3-3/2.
\end{align*}
A direct calculation shows that the operators $J_i$ satisfy the commutation relations
\begin{align}
\label{Sym-Alg}
 [J_i,J_j]=i\epsilon_{ijk}\Big(J_{k}+2\mu_k\,(\Gamma+1)\,\sigma_k R_k+2\mu_i\mu_j\,\sigma_k R_iR_j\Big).
\end{align}
The operator $\Gamma$ also admits the three involutions
\begin{align}
\label{Discrete-Sym}
 Z_i=\sigma_i R_i,\qquad Z_i^2=1,\qquad i=1,2,3,
\end{align}
as symmetry operators, i.e
\begin{align*}
 [\Gamma,Z_i]=0,\qquad i=1,2,3.
\end{align*}
The commutation relations between $Z_i$ and $J_i$ read
\begin{align}
\label{Com-Z}
[J_i,Z_i]=0,\qquad \{J_i, Z_j\}=0,\qquad \{Z_i,Z_j\}=0,\qquad i\neq j.
\end{align}

The involutions \eqref{Discrete-Sym} and the relations \eqref{Com-Z} can be exploited to give another presentation of the symmetry algebra of $\Gamma$. Let $K_i$, $i=1,2,3$, be defined as follows
\begin{align}
\label{BI-Operators}
 K_i=-i\,J_i\,Z_{j}\,Z_{k},
\end{align}
where $(ijk)$ is again a cyclic permutation of $\{1,2,3\}$. Since the operators $J_i$ and $Z_i$ both commute with $\Gamma$, it follows that the operators $K_i$ also commute with $\Gamma$. Upon combining the relations \eqref{Sym-Alg} and \eqref{Com-Z}, one finds that the symmetries $K_i$ satisfy the commutation relations
\begin{align}
\label{BI-1}
 \{K_i,K_j\}&=\epsilon_{ijk}\Big(K_{k}+2\mu_k\,(\Gamma+1)\,R_1R_2R_3+2\mu_i\mu_j\Big).
\end{align}
The invariance algebra \eqref{BI-1} of the Dirac--Dunkl operator thus has the form of the Bannai--Ito algebra. More precisely, \eqref{BI-1} can be viewed as a central extension of the Bannai--Ito algebra given the presence of the central element $(\Gamma+1)R_1R_2R_3$ on the right hand side. In terms of the symmetries $K_i$, the $\Gamma$ operator reads
\begin{align}
\label{Gamma-k}
 \Gamma=K_1R_2R_3+K_2R_1R_3+K_3 R_1 R_2-\mu_1 R_1-\mu_2 R_2-\mu_3 R_3-3/2.
\end{align}
The commutation relations between the symmetries $K_i$ and the involutions $Z_i$ are
\begin{align*}
 [K_i,Z_j]=0.
\end{align*}
The Casimir operator of the Bannai--Ito algebra, denoted by $Q$,  has the expression \cite{2012_Tsujimoto&Vinet&Zhedanov_AdvMath_229_2123}
\begin{align}
\label{BI-Casimir}
 Q=K_1^2+K_2^2+K_3^2.
\end{align}
It is easily verified that $Q$ commutes with $K_i$ and $Z_i$ for $i=1,2,3$. A direct calculation shows that in the realization \eqref{BI-Operators}, the Casimir operator \eqref{BI-Casimir} can be written as
\begin{align*}
 Q=(\Gamma+1)^2+\mu_1^2+\mu_2^2+\mu_3^2-1/4.
\end{align*}
It follows from \eqref{Gamma-Eigenvalues} and \eqref{BI-1} that the space of Dunkl monogenics $\mathcal{M}_{N}(\mathbb{R}^3)$ of degree $N$ carries representations of the Bannai--Ito algebra \eqref{BI-Algebra}. The precise content of $\mathcal{M}_{N}(\mathbb{R}^3)$ in representations of the Bannai--Ito algebra will be determined in section 5.
\section{Representations of the Bannai--Ito algebra}
In this section, the representations of the Bannai--Ito algebra corresponding to the realization \eqref{BI-1} are constructed using ladder operators.

On the space of Dunkl monogenics $\mathcal{M}_{N}(\mathbb{R}^3)$ of degree $N$, the symmetries $K_i$ of the Dirac--Dunkl operator satisfy the commutation relations
\begin{align}
\label{BI-2}
 \{K_1,K_2\}=K_3+\omega_3,\quad \{K_2,K_3\}=K_1+\omega_1,\quad \{K_3,K_1\}=K_2+\omega_2,
\end{align}
with structure constants
\begin{align}
\label{Structure}
 \omega_3=2\mu_1\mu_2+2\mu_3\mu_N,\quad \omega_1=2\mu_2\mu_3+2\mu_1\mu_N,\quad \omega_2=2\mu_3\mu_1+2\mu_2\mu_N,
\end{align}
and where we have defined
\begin{align}
\label{Mu-N}
 \mu_{N}=(-1)^{N}(N+\mu_1+\mu_2+\mu_3+1).
\end{align}
On $\mathcal{M}_{N}(\mathbb{R}^3)$, the Casimir operator \eqref{BI-Casimir} takes the value
\begin{align}
\label{Casimir-Value}
Q=(N+\mu_1+\mu_2+\mu_3+1)^2+\mu_1^2+\mu_2^2+\mu_3^2-1/4\equiv q_{N}.
\end{align}
We seek to construct the representations of the Bannai--Ito algebra \eqref{BI-2} on the space spanned by the orthonormal basis vectors $\ket{N,k}$ characterized by the eigenvalue equations
\begin{align}
\label{Eigen-Setup}
 K_3\;\ket{N,k}=\lambda_k\;\ket{N,k},\qquad Q\;\ket{N,k}=q_{N}\;\ket{N,k}.
\end{align}
Since the operators $K_i$ are potential observables, it is formally assumed that
\begin{align}
\label{Star-Condition}
 K_i^{\dagger}=K_i,\qquad i=1,2,3.
\end{align}
To characterize the representation, one needs to determine the spectrum of $K_3$ and the action of the operator $K_1$ on the basis vectors $\ket{N,k}$. 

Introduce the operators $K_{+}$ and $K_{-}$ defined by \cite{2012_Tsujimoto&Vinet&Zhedanov_AdvMath_229_2123}
\begin{align}
\label{Ladder-Def}
\begin{aligned}
 K_{+}&=(K_1+K_2)(K_3-1/2)-(\omega_1+\omega_2)/2,
 \\
 K_{-}&=(K_1-K_2)(K_3+1/2)+(\omega_1-\omega_2)/2.
 \end{aligned}
\end{align}
Using the defining relations \eqref{BI-2} and the Hermiticity condition \eqref{Star-Condition}, it is seen that the operators $K_{\pm}$ are skew-Hermitian, i.e. 
\begin{align*}
 K_{\pm}^{\dagger}=-K_{\pm}.
\end{align*}
Moreover, a direct calculation shows that they satisfy the commutation relations
\begin{align}
\label{Commu}
 \{K_3,K_{+}\}=K_{+},\quad \{K_{3},K_{-}\}=-K_{-},
\end{align}
whence it follows that
\begin{align*}
 [K_3,K_{+}^2]=0,\qquad [K_3,K_{-}^2]=0.
\end{align*}
Using \eqref{Commu}, one can write
\begin{align}
\label{Ultra}
\begin{aligned}
 K_3\,K_{+}\,\ket{N,k}&=(K_{+}-K_{+}K_3)\,\ket{N,k}=(1-\lambda_k)\,K_{+}\,\ket{N,k},
 \\
 K_3\,K_{-}\,\ket{N,k}&=(-K_{-}-K_{-}K_{3})\,\ket{N,k}=(-1-\lambda_k)\,K_{-}\,\ket{N,k}.
 \end{aligned}
\end{align}
The above relations indicate that $K_{+} \ket{N,k}$ and $K_{-}\ket{N,k}$ are eigenvectors of $K_3$ with eigenvalues $(1-\lambda_k)$ and $-(1+\lambda_k)$, respectively. One has the two inequalities
\begin{subequations}
\begin{align}
\label{First}
 ||K_{+}\ket{N,k}||^2&=\Braket{N,k}{\,K_{+}^{\dagger}\,K_{+}\,}{N,k}\geqslant 0,
 \\
 \label{Second}
 ||K_{-}\ket{N,k}||^2&=\Braket{N,k}{\,K_{-}^{\dagger}\,K_{-}\,}{N,k}\geqslant 0.
\end{align}
\end{subequations}
Consider the LHS of \eqref{First}. The operator $K_{+}^{\dagger}$ is of the form
\begin{align*}
 K_{+}^{\dagger}=(K_3-1/2)(K_1+K_2)-(\omega_1+\omega_2)/2.
\end{align*}
Upon using the commutation relations \eqref{BI-2} and \eqref{Casimir-Value}, it is seen that
\begin{align}
\label{Kp-KpDagger}
 K_{+}^{\dagger}K_{+}=(K_3-1/2)^2(Q-K_3^2+K_3+\omega_3)-(\omega_1+\omega_2)^2/4.
\end{align}
Using the above expression in \eqref{First} with \eqref{Structure} and \eqref{Casimir-Value}, one finds a factorized form for the inequality
\begin{align}
\label{Ineq1}
 -(\mu_1+\mu_2+1/2-\lambda_k)(\mu_{N}+\mu_3+1/2-\lambda_k)(\lambda_k+\mu_1+\mu_2-1/2)(\lambda_k+\mu_N+\mu_3-1/2)\geqslant 0.
\end{align}
which is equivalent to
\begin{align}
\label{Ineq11}
 \begin{cases}
  \mu_1+\mu_2\leqslant|\lambda_k-1/2|\leqslant N+\mu_1+\mu_2+2\mu_3+1, & \text{$N$ even,}
  \\
  \mu_1+\mu_2\leqslant |\lambda_k-1/2|\leqslant N+\mu_1+\mu_2+1,& \text{$N$ odd.}
 \end{cases}
\end{align}
Proceeding similarly for $K_{-}$, one can write
\begin{align}
\label{Km-KmDagger}
  K_{-}^{\dagger}K_{-}=(K_3+1/2)^2(Q-K_3^2-K_3-\omega_3)-(\omega_1-\omega_2)^2/4.
\end{align}
and one finds that \eqref{Second} amounts to
\begin{align}
\label{Ineq2}
 -(\mu_1-\mu_2-1/2-\lambda_k)(\mu_{3}-\mu_N-1/2-\lambda_k)(\lambda_k+\mu_1-\mu_2+1/2)(\lambda_k+\mu_3-\mu_N+1/2)\geqslant 0,
\end{align}
which can be expressed as
\begin{align}
\label{Ineq22}
 \begin{cases}
  |\mu_1-\mu_2|\leqslant|\lambda_k+1/2|\leqslant N+\mu_1+\mu_2+1, & \text{$N$ even},
  \\
  |\mu_1-\mu_2|\leqslant|\lambda_k+1/2|\leqslant N+\mu_1+\mu_2+2\mu_3+1, & \text{$N$ odd},
 \end{cases}
\end{align}
Upon combining \eqref{Ineq11} and \eqref{Ineq22}, we choose
\begin{align*}
 \lambda_0=\mu_1+\mu_2+1/2,
\end{align*}
whence it follows from \eqref{Ineq1} that 
\begin{align*}
 K_{+}\ket{N,0}=0.
\end{align*}
Let us mention that the other choices $\lambda_0=-(\mu_1+\mu_2+1/2)$ and $\widetilde{\lambda}_0=\pm(-\mu_1-\mu_2+1/2)$ permitted by \eqref{Ineq11} do not lead to admissible representations. 

Starting from the vector $\ket{N,0}$ with eigenvalue $\lambda_0$, one can obtain a string of eigenvectors of $K_3$ with different eigenvalues by successively applying $K_{+}$ and $K_{-}$. The eigenvalues
\begin{align}
\label{K3-Eigenvalues}
 \lambda_k=(-1)^{k}(k+\mu_1+\mu_2+1/2),\qquad k=0,1,2,3,\ldots.
\end{align}
are obtained by applying $K_3$ on the vectors 
\begin{align}
\label{Sequence}
 \ket{N,0},\quad K_{-}\ket{N,0},\quad K_{+}K_{-}\ket{N,0},\quad K_{-}K_{+}K_{-}\ket{N,0}, \ldots
\end{align}
One needs to alternate the application of $K_{+}$ and $K_{-}$ since $K_{\pm}^2$ commute with $K_3$ and hence their action does not produce an eigenvector with a different eigenvalue. Using \eqref{K3-Eigenvalues}, one can write
\begin{gather*}
 ||K_{+}\ket{N,k}||^2=
 \begin{cases}
  \rho_{k}^{(N)}, & \text{$k$ even,}
  \\
  \rho_{k+1}^{(N)},& \text{$k$ odd,}
 \end{cases}
\end{gather*}
where
\begin{align}
\label{Rho}
 \rho_k^{(N)}=-k(k+2\mu_1+2\mu_2)(k+\mu_1+\mu_2+\mu_3+\mu_N)(k+\mu_1+\mu_2-\mu_3-\mu_{N}),
\end{align}
and also
\begin{align*}
 ||K_{-}\ket{N,k}||^2=
 \begin{cases}
  \sigma_{k+1}^{(N)}, & \text{$k$ even,}
  \\
  \sigma_{k}^{(N)},& \text{$k$ odd,}
 \end{cases}
\end{align*}
with
\begin{align}
\label{Sigma}
 \sigma_k^{(N)}=-(k+2\mu_1)(k+2\mu_2)(k+\mu_1+\mu_2-\mu_3+\mu_{N})(k+\mu_1+\mu_2+\mu_3-\mu_N).
\end{align}
It is verified that the positivity conditions $\rho_{k}^{(N)}>0$ and $\sigma_{k}^{(N)}>0$ are satisfied for all $k=0,1,\ldots,N$, provided that $\mu_i\geqslant 0$ for $i=1,2,3$. Following \eqref{Sequence}, \eqref{Rho} and \eqref{Sigma}, we define the orthonormal basis vectors $\ket{N,k}$ from $\ket{N,0}$ as follows:
\begin{align}
\label{State-Def}
 \ket{N,k+1}=
\begin{cases}
 \frac{1}{\sqrt{||K_{-}\ket{N,k}||^2}}\;K_{-}\ket{N,k}, & \text{$k$ even},
 \\
 \frac{-1}{\sqrt{||K_{+}\ket{N,k}||^2}}\;K_{+}\ket{N,k}, & \text{$k$ odd},
\end{cases}
\end{align}
where the phase factor was chosen to ensure the condition $K_{\pm}^{\dagger}=-K_{\pm}$. From \eqref{Kp-KpDagger}, \eqref{Km-KmDagger}, \eqref{Sequence} and \eqref{State-Def}, the actions of the ladder operators $K_{\pm}$ are seen to have the expressions
\begin{align}
\label{Actions}
\begin{aligned}
 K_{+}\ket{N,k}=
 \begin{cases}
  \sqrt{\rho_{k}^{(N)}}\;\ket{N,k-1}, & \text{$k$ even},
  \\
  -\sqrt{\rho_{k+1}^{(N)}}\;\ket{N,k+1}, & \text{$k$ odd},
 \end{cases}
 \\[.1cm]
 K_{-}\ket{N,k}=
\begin{cases}
  \sqrt{\sigma_{k+1}^{(N)}}\,\ket{N,k+1},& \text{$k$ even},
  \\
   -\sqrt{\sigma_{k}^{(N)}}\,\ket{N,k-1}, & \text{$k$ odd}.
 \end{cases}
 \end{aligned}
\end{align}
As is observed in \eqref{Rho} and \eqref{Sigma}, one has $K_{+}\ket{N,N}=0$ when $N$ is odd and $K_{-}\ket{N,N}=0$ when $N$ is even. As a result, the representation has dimension $N+1$. Moreover, it immediately follows from the actions \eqref{Actions} that the representation is irreducible, as there are no invariant subspaces.

Let us now give the actions of the generators. The eigenvalues of $K_3$ are of the form
\begin{align*}
 K_3\ket{N,k}=(-1)^{k}(k+\mu_1+\mu_2+1/2)\,\ket{N,k},\qquad k=0,1,\ldots,N.
\end{align*}
The action of the operator $K_1$ in the basis $\ket{N,k}$ can be obtained directly from the definitions \eqref{Ladder-Def} and the actions \eqref{Actions}. One finds that $K_1$ acts in the tridiagonal fashion
\begin{align*}
 K_1\ket{N,k}=U_{k+1}\ket{N,k+1}+V_{k}\ket{N,k}+U_{k}\ket{N,k-1},
\end{align*}
with
\begin{align*}
 U_{k}=\sqrt{A_{k-1}C_{k}},\quad V_{k}=\mu_2+\mu_3+1/2-A_{k}-C_{k},
\end{align*}
where the coefficients $A_{k}$ and $C_{k}$ read
\begin{align}
\begin{aligned}
 A_{k}&=
 \begin{cases}
  \frac{(k+2\mu_2+1)(k+\mu_1+\mu_2+\mu_3-\mu_{N}+1)}{2(k+\mu_1+\mu_2+1)}, & \text{$k$ even},
  \\
  \frac{(k+2\mu_1+2\mu_2+1)(k+\mu_1+\mu_2+\mu_3+\mu_{N}+1)}{2(k+\mu_1+\mu_2+1)}, & \text{$k$ odd},
 \end{cases}
 \\
 C_{k}&=
 \begin{cases}
  -\frac{k(k+\mu_1+\mu_2-\mu_3-\mu_{N})}{2(k+\mu_1+\mu_2)}, & \text{$k$ even},
  \\
  -\frac{(k+2\mu_1)(k+\mu_1+\mu_2-\mu_3+\mu_{N})}{2(k+\mu_1+\mu_2)}, & \text{$k$ odd}.
 \end{cases} 
 \end{aligned}
\end{align}
For $\mu_i\geqslant0$, $i=1,2,3$, one has $U_\ell>0$ for $\ell=1,\ldots,N$ and $U_0=U_{N+1}=0$. Hence in the basis $\ket{N,k}$, the operator $K_1$ is represented by a symmetric $(N+1)\times (N+1)$ matrix.

It is observed that the commutation relations \eqref{BI-2} along with the structure constants \eqref{Structure} and the Casimir value \eqref{Casimir-Value} are invariant under any cyclic permutation of the pairs $(K_i,\mu_i)$ for $i=1,2,3$. Consequently, the matrix elements of the generators in other bases, for example bases in which $K_1$ or $K_2$ are diagonal, can be obtained directly by applying the corresponding cyclic permutation on the parameters $\mu_i$.

\section{Eigenfunctions of the spherical Dirac--Dunkl operator}
In this section, a basis for the space of Dunkl monogenics $\mathcal{M}_{N}(\mathbb{R}^3)$ of degree $N$ is constructed using a Cauchy-Kovalevskaia extension theorem. It is shown that the basis functions transform irreducibly under the action of the Bannai--Ito algebra. The wavefunctions are shown to be orthogonal with respect to a scalar product defined as an integral over the 2-sphere.

\subsection{Cauchy-Kovalevskaia map}
Let $\widetilde{\underline{D}}$, $\widetilde{\underline{x}}$ and $\widetilde{\mathbb{E}}$ be defined as follows:
\begin{align*}
 \widetilde{\underline{D}}=\sigma_1 T_1+\sigma_2 T_2,\qquad \widetilde{\underline{x}}=\sigma_1 x_1+\sigma_2 x_2,\quad \widetilde{\mathbb{E}}=x_1 \partial_{x_1}+x_2 \partial_{x_2}.
\end{align*}
There is an isomorphism $\mathbf{CK}_{x_3}^{\mu_3}: \mathcal{P}_{N}(\mathbb{R}^2)\otimes \mathbb{C}^2\longrightarrow \mathcal{M}_{N}(\mathbb{R}^3)$, between the space of spinor-valued homogeneous polynomials of degree $N$ in the variables $(x_1,x_2)$ and the space of Dunkl monogenics of degree $N$ in the variables $(x_1,x_2,x_3)$.
\begin{Proposition}
 The isomorphism $\mathbf{CK}_{x_3}^{\mu_3}$ between $\mathcal{P}_{N}(\mathbb{R}^2)\otimes \mathbb{C}^2$ and $\mathcal{M}_{N}(\mathbb{R}^3)$ has the explicit expression
 \begin{align}
 \label{CK-X3}
  \mathbf{CK}_{x_3}^{\mu_3}={}_0F_{1}\left(\genfrac{}{}{0pt}{}{-}{\mu_3+1/2}\;\Big\rvert\;-\left(\frac{x_3\,\widetilde{\underline{D}}}{2}\right)^2\right)
  -\frac{\sigma_3\,x_3\,\widetilde{\underline{D}}}{2\mu_3+1}\;\;{}_0F_{1}\left(\genfrac{}{}{0pt}{}{-}{\mu_3+3/2}\;\Big\rvert\;-\left(\frac{x_3\,\widetilde{\underline{D}}}{2}\right)^2\right),
 \end{align}
 where ${}_pF_{q}$ is the generalized hypergeometric series \cite{2001_Andrews&Askey&Roy}.
\end{Proposition}
\begin{proof}
 Let $p(x_1,x_2)\in \mathcal{P}_{n}(\mathbb{R}^2)\otimes \mathbb{C}^2$. We set 
 \begin{align*}
  \mathbf{CK}_{x_3}^{\mu_3}[p(x_1,x_2)]=\sum_{\alpha=0}^{n}(\sigma_3 x_3)^{\alpha}p_{\alpha}(x_1,x_2),
 \end{align*}
with $p_0(x_1,x_2)\equiv p(x_1,x_2)$ and $p_{\alpha}(x_1,x_2)\in \mathcal{P}_{n-\alpha}(\mathbb{R}^2)\otimes \mathbb{C}^2$ and we determine the $p_{\alpha}(x_1,x_2)$ such that $\mathbf{CK}_{x_3}^{\mu_3}[p(x_1,x_2)]$ is in the kernel of $\underline{D}$. One has
\begin{multline*}
 \underline{D}\,\mathbf{CK}_{x_3}^{\mu_3}[p(x_1,x_2)]=
 \sum_{\alpha=0}^{n}(-\sigma_3 x_3)^{\alpha}(\sigma_1 T_1+\sigma_2 T_2)p_{\alpha}(x_1,x_2)+\sum_{\alpha=1}^{n}\sigma_3^{\alpha+1} (T_3 x_3^{\alpha}) p_{\alpha}(x_1,x_2)
 \\
 =\sum_{\alpha=0}^{n}(-\sigma_3 x_3)^{\alpha}(\sigma_1 T_1+\sigma_2 T_2)p_{\alpha}(x_1,x_2)
 +\sum_{\alpha=1}^{n}\sigma_3^{\alpha+1}\;[\alpha+\mu_3 (1-(-1)^{\alpha})]\;x_3^{\alpha-1}p_{\alpha}(x_1,x_2).
\end{multline*}
Imposing the condition $\underline{D}\,\mathbf{CK}_{x_3}^{\mu_3}[p(x_1,x_2)]=0$ leads to the equations
\begin{align*}
 \sum_{\alpha=0}^{n}(-1)^{\alpha+1}(\sigma_3 x_3)^{\alpha}\,(\sigma_1 T_1+\sigma_2 T_2)p_{\alpha}(x_1,x_2)=
 \sum_{\alpha=0}^{n-1}(\sigma_3 x_3)^{\alpha}[\alpha+\mu_3(1+(-1)^{\alpha})]p_{\alpha+1}(x_1,x_2),
\end{align*}
from which one finds that
\begin{align*}
 p_{2\alpha}(x_1,x_2)&=\left[\frac{(-1)^{\alpha}}{2^{2\alpha}\,\alpha!\,(\mu_3+1/2)_{\alpha}}\right](\sigma_1 T_1+\sigma_2 T_2)^{2\alpha}p(x_1,x_2),
 \\
 p_{2\alpha+1}(x_1,x_2)&=\left[\frac{(-1)^{\alpha+1}}{2^{2\alpha+1}\,\alpha!\,(\mu_3+1/2)\,(\mu_3+3/2)_{\alpha}}\right](\sigma_1 T_1+\sigma_2 T_2)^{2\alpha+1}p(x_1,x_2),
\end{align*}
where $(a)_{n}$ stands for the Pochhammer symbol. It is seen that the above corresponds to the hypergeometric expression \eqref{CK-X3}.
\end{proof}
The inverse of the isomorphism $\mathbf{CK}_{x_3}^{\mu_3}$ is clearly given by $I_{x_3}$ with $I_{x_3}f(x_1,x_2,x_3)=f(x_1,x_2,0)$. When $\mu_3=0$, the operator $\mathbf{CK}_{x_3}^{\mu_3}$ reduces to the well-known Cauchy-Kovalevskaia extension operator for the standard Dirac operator, as determined in \cite{1982_Brackx&Delanghe&Sommen}. It is manifest that proposition 1 can be extended to any dimension. Thus, in a similar fashion, one has the isomorphism
\begin{align*}
 \mathbf{CK}_{x_2}^{\mu_2}: \mathcal{P}_{k}(\mathbb{R})\otimes \mathbb{C}^2\longrightarrow \mathcal{M}_{k}(\mathbb{R}^2),
\end{align*}
between the space of spinor-valued homogeneous polynomials in the variable $x_1$ and the space of Dunkl monogenics of degree $k$ in the variables $(x_1,x_2)$. This isomorphism has the explicit expression
\begin{align}
\label{CK-X2}
 \mathbf{CK}_{x_2}^{\mu_2}={}_0F_{1}\left(\genfrac{}{}{0pt}{}{-}{\mu_2+1/2}\;\Big\rvert\;-\left(\frac{x_2\,\sigma_1 T_1}{2}\right)^2\right)-\frac{\sigma_2\,x_2\,(\sigma_1 T_1)}{2\mu_2+1}\;{}_0F_{1}\left(\genfrac{}{}{0pt}{}{-}{\mu_2+3/2}\;\Big\rvert\;-\left(\frac{x_2\,\sigma_1 T_1}{2}\right)^2\right).
\end{align}
\subsection{A basis for $\mathcal{M}_{N}(\mathbb{R}^3)$}
Let us now show how a basis for the space of Dunkl monogenics of degree $N$ in $\mathbb{R}^3$ can be constructed using the $\mathbf{CK}_{x_i}^{\mu_i}$ extension operators and the Fischer decomposition theorem \eqref{Fischer-Decomposition}. Let $\chi_{+}=(1,0)^{\top}$ and $\chi_{-}=(0,1)^{\top}$ denote the basis spinors; one has $\mathbb{C}^2=\mathrm{Span}\{\chi_{\pm}\}$. Consider the following tower of $\mathbf{CK}$ extensions and Fischer decompositions:
\begin{center}
\begin{equation*}
\xymatrix@=12pt{
&&& \mathcal{P}_{N}(\mathbb{R}^2)\otimes \mathbb{C}^2 \ar[rr]_-{\mathbf{CK}_{x_3}^{\mu_3}}&& \mathcal{M}_{N}(\mathbb{R}^3) 
\\
\mathrm{Span}\{\,x_1^{N}\;\chi_{\pm}\,\}=\mathcal{P}_{N}(\mathbb{R})\otimes \mathbb{C}^2 \ar[rr]_-{\mathbf{CK}_{x_2}^{\mu_2}}&& \mathcal{M}_{N}(\mathbb{R}^2) \ar@{~>}[r]& \mathcal{M}_{N}(\mathbb{R}^2) \ar@{}[u]_{\parallel}&& 
\\
\mathrm{Span}\{\,x_1^{N-1}\;\chi_{\pm}\,\}=\mathcal{P}_{N-1}(\mathbb{R})\otimes \mathbb{C}^2 \ar[rr]_-{\mathbf{CK}_{x_2}^{\mu_2}}&&  \mathcal{M}_{N-1}(\mathbb{R}^2) \ar@{~>}[r]& \widetilde{\underline{x}}\mathcal{M}_{N-1}(\mathbb{R}^2) \ar@{}[u]_{\bigoplus}&& 
\\
\vdots && & \;\,\quad \vdots\ar@{}[u]_{\bigoplus}  &&
\\
\mathrm{Span}\{\,x_1^{k}\;\chi_{\pm}\,\}=\mathcal{P}_{k}(\mathbb{R})\otimes \mathbb{C}^2 \ar[rr]_-{\mathbf{CK}_{x_2}^{\mu_2}}&&  \mathcal{M}_{k}(\mathbb{R}^2) \ar@{~>}[r]& \widetilde{\underline{x}}^{N-k}\mathcal{M}_{k}(\mathbb{R}^2) \ar@{}[u]_{\bigoplus} &&  \sim\psi_{k,\pm}^{(N)}
\\
\vdots && & \;\,\quad \vdots\ar@{}[u]_{\bigoplus} &&
\\
\mathrm{Span}\{\,x_1\;\chi_{\pm}\,\}=\mathcal{P}_{1}(\mathbb{R})\otimes \mathbb{C}^2 \ar[rr]_-{\mathbf{CK}_{x_2}^{\mu_2}} &&  \mathcal{M}_{1}(\mathbb{R}^2) \ar@{~>}[r] & \widetilde{\underline{x}}^{N-1}\mathcal{M}_{1}(\mathbb{R}^2) \ar@{}[u]_{\bigoplus}&& 
\\
\mathrm{Span}\{\chi_{\pm}\}=\mathcal{P}_{0}(\mathbb{R})\otimes \mathbb{C}^2 \ar[rr]_-{\mathbf{CK}_{x_2}^{\mu_2}}&&  \mathcal{M}_{0}(\mathbb{R}^2) \ar@{~>}[r]& \widetilde{\underline{x}}^{N}\mathcal{M}_{0}(\mathbb{R}^2) \ar@{}[u]_{\bigoplus}&& 
}
\end{equation*}
Diagram 1. Horizontally, application of the $\mathbf{CK}$ map and multiplication by $\widetilde{\underline{x}}$. Vertically, Fischer decomposition theorem for $\mathcal{P}_{N}(\mathbb{R}^2)\otimes \mathbb{C}^2$. 
\end{center}
As can be seen from the above diagram, the spinors
\begin{align}
\label{Basis}
 \psi_{k,\pm}^{(N)}=\mathbf{CK}_{x_3}^{\mu_3}\Big[\widetilde{\underline{x}}^{N-k}\;\mathbf{CK}_{x_2}^{\mu_2}\big[x_1^{k}\big]\Big]\chi_{\pm},\qquad k=0,1,\ldots,N,
\end{align}
provide a basis for the space of Dunkl monogenics of degree $N$ in $(x_1,x_2,x_3)$. The basis spinors \eqref{Basis} can be calculated explicitly. To perform the calculation, one needs the identities
\begin{align}
\label{Formulas}
\begin{aligned}
 \widetilde{\underline{D}}^{2\alpha}\,\widetilde{\underline{x}}^{2\beta}\,M_k&=2^{2\alpha}(-\beta)_{\alpha}(1-k-\beta-\gamma_2)_{\alpha}\;\widetilde{\underline{x}}^{2\beta-2\alpha}\,M_{k},
 \\
  \widetilde{\underline{D}}^{2\alpha+1}\,\widetilde{\underline{x}}^{2\beta}\,M_k&=\beta\;2^{2\alpha+1}(1-\beta)_{\alpha}(1-k-\beta-\gamma_2)_{\alpha}\;\widetilde{\underline{x}}^{2\beta-2\alpha-1}M_{k},
  \\
  \widetilde{\underline{D}}^{2\alpha}\,\widetilde{\underline{x}}^{2\beta+1}\,M_k&=2^{2\alpha}\;(-\beta)_{\alpha}\;(-k-\beta-\gamma_2)_{\alpha}\;\widetilde{\underline{x}}^{2\beta-2\alpha+1}M_{k},
  \\
    \widetilde{\underline{D}}^{2\alpha+1}\,\widetilde{\underline{x}}^{2\beta+1}\,M_k&=(k+\beta+\gamma_2)\;2^{2\alpha+1}\,(-\beta)_{\alpha}(1-k-\beta-\gamma_2)_{\alpha}\;\widetilde{\underline{x}}^{2\beta-2\alpha}\,M_{k},
    \end{aligned}
\end{align}
where $M_{k}\in \mathcal{M}_{k}(\mathbb{R}^2)$ and $\gamma_2=\mu_1+\mu_2+1$. The formulas \eqref{Formulas}, given in \cite{2012_DeBie&Orsted&Somberg&Soucek_TransAmerMathSoc_364_3875} for arbitrary dimension, are easily obtained from the relations
\begin{align*}
 \widetilde{\underline{D}}\;\widetilde{\underline{x}}^{2\beta} M_{k}=2\beta\;\widetilde{\underline{x}}^{2\beta-1} M_{k},\qquad \widetilde{\underline{D}}\,\widetilde{\underline{x}}^{2\beta+1}M_{k}=2(\beta+k+\gamma_2)\,\widetilde{\underline{x}}^{2\beta} M_{k},
\end{align*}
which follow from the commutation relations
\begin{align}
\label{Commu-2}
 [\widetilde{\underline{D}}, \widetilde{\underline{x}}^2]=2 \widetilde{\underline{x}},\quad \{\widetilde{\underline{D}},\widetilde{\underline{x}}\}=2(\widetilde{\mathbb{E}}+\gamma_2).
\end{align}
Similar formulas hold in the one-dimensional case. To present the result, we shall need the Jacobi polynomials $P_{n}^{(\alpha,\beta)}(x)$, defined as \cite{2010_Koekoek_&Lesky&Swarttouw}
\begin{align*}
 P_{n}^{(\alpha,\beta)}(x)=\frac{(\alpha+1)_{n}}{n!}\;
 {}_2F_{1}\left(
 \genfrac{}{}{0pt}{}{-n,n+\alpha+\beta+1}{\alpha+1}
 \;\Big \rvert
 \;
 \frac{1-x}{2}
 \right).
\end{align*}
The following identity:
\begin{align*}
 (x+y)^{m}P_{m}^{(\alpha,\beta)}\left(\frac{x-y}{x+y}\right)=\frac{(\alpha+1)_{m}}{m!}\,x^{m}\,
 {}_2F_{1}\left(
 \genfrac{}{}{0pt}{}{-m,-m-\beta}{\alpha+1}
 \;\Big \rvert
 \;
 -\frac{y}{x}
 \right),
\end{align*}
will also be needed.

Computing \eqref{Basis} using the definitions \eqref{CK-X3}, \eqref{CK-X2}, the formulas \eqref{Formulas} and the above identity, a long but otherwise straightforward calculation shows that the basis spinors have the expression
\begin{align}
\label{psi}
 \psi_{k,\pm}^{(N)}=q_{N-k}(x_3,\widetilde{\underline{x}})\,m_{k}(x_2,x_1)\,\chi_{\pm},\qquad k=0,\ldots,N,
\end{align}
where
\begin{align*}
 m_{k}(x_1,x_2)=\mathbf{CK}_{x_2}^{\mu_2}[x_1^{k}].
\end{align*}
One has
\begin{multline}
\label{q}
 q_{N-k}(x_3,\widetilde{\underline{x}})=\frac{\beta!}{(\mu_3+1/2)_{\beta}}\,(x_1^2+x_2^2+x_3^2)^{\beta}\,
 \\ \times
 \begin{cases}
  P_{\beta}^{(\mu_3-1/2,k+\mu_1+\mu_2)}\left(\frac{x_1^2+x_2^2-x_3^2}{x_1^2+x_2^2+x_3^2}\right)
  & \text{$N-k=2\beta$},
  \\[.2cm]
  \qquad \qquad \qquad -\frac{\sigma_3 x_3 \widetilde{\underline{x}}}{x_1^2+x_2^2+x_3^2}\;P_{\beta-1}^{(\mu_3+1/2,k+\mu_1+\mu_2+1)}\left(\frac{x_1^2+x_2^2-x_3^2}{x_1^2+x_2^2+x_3^2}\right),
  \\
  \\
  \widetilde{\underline{x}}\;P_{\beta}^{(\mu_3-1/2,k+\mu_1+\mu_2+1)}\left(\frac{x_1^2+x_2^2-x_3^2}{x_1^2+x_2^2+x_3^2}\right)
  ,
  &
  \text{$N-k=2\beta+1$},
  \\[.2cm]
  \qquad \qquad \qquad 
  -\sigma_3 x_3\left(\frac{k+\beta+\mu_1+\mu_2+1}{\beta+\mu_3+1/2}\right)\;P_{\beta}^{(\mu_3+1/2,k+\mu_1+\mu_2)}\left(\frac{x_1^2+x_2^2-x_3^2}{x_1^2+x_2^2+x_3^2}\right),
 \end{cases}
\end{multline}
and
\begin{multline}
\label{m}
 m_{k}(x_2,x_1)=\frac{\beta!}{(\mu_2+1/2)_{\beta}}(x_1^2+x_2^2)^{\beta}
 \\
 \times
 \begin{cases}
  P_{\beta}^{(\mu_2-1/2,\mu_1-1/2)}\left(\frac{x_1^2-x_2^2}{x_1^2+x_2^2}\right)-\frac{\sigma_2 x_2\sigma_1 x_1}{x_1^2+x_2^2}\,P_{\beta-1}^{(\mu_2+1/2,\mu_1+1/2)}\left(\frac{x_1^2-x_2^2}{x_1^2+x_2^2}\right), & \text{$k=2\beta$},
  \\
  x_1\,P_{\beta}^{(\mu_2-1/2,\mu_1+1/2)}\left(\frac{x_1^2-x_2^2}{x_1^2+x_2^2}\right)
  -\sigma_2\,x_2\,\sigma_1 \left(\frac{\beta+\mu_1+1/2}{\beta+\mu_2+1/2}\right)\,P_{\beta}^{(\mu_2+1/2,\mu_1-1/2)}\left(\frac{x_1^2-x_2^2}{x_1^2+x_2^2}\right),
  & \text{$k=2\beta+1$}.
 \end{cases}
\end{multline}
\subsection{Basis spinors and representations of the Bannai--Ito algebra}
The basis vectors $\psi_{k,\pm}^{(N)}$ transform irreducibly under the action of the Bannai--Ito algebra. This can be established as follows. By construction, $\psi_{k,\pm}^{(N)}\in \mathcal{M}_{N}(\mathbb{R}^3)$, and thus \eqref{Gamma-Eigenvalues} gives
\begin{align*}
 (\Gamma+1)\;\psi_{k,\pm}^{(N)}=(N+\mu_1+\mu_2+\mu_3+1)\;\psi_{k,\pm}^{(N)}.
\end{align*}
Hence we have
\begin{align}
\label{QQ}
 Q\psi_{k,\pm}^{(N)}=\left((\Gamma+1)^2+\mu_1^2+\mu_2^2+\mu_3^2-1/4\right)=q_{N}\,\psi_{k,\pm}^{(N)},
\end{align}
as in \eqref{Eigen-Setup}. The spinors \eqref{psi} are also eigenvectors of $K_3$. To prove this result, one first observes that $K_3$ can be written as
\begin{align*}
 K_3=-\frac{1}{2}\left([\widetilde{\underline{x}},\widetilde{\underline{D}}]+1\right)R_1R_2.
\end{align*}
Since $K_3$ acts only on the variables $(x_1,x_2)$ and since $[K_3,\widetilde{\underline{x}}]=0$, one has
\begin{align*}
 K_3\;\psi_{k,\pm}^{(N)}&=K_3\;\mathbf{CK}_{x_3}^{\mu_3}\Big[\widetilde{\underline{x}}^{N-k}\;\mathbf{CK}_{x_2}^{\mu_2}\big[x_1^{k}\big]\Big]\chi_{\pm}
 = \mathbf{CK}_{x_3}^{\mu_3}\Big[\widetilde{\underline{x}}^{N-k}\;K_3\,\mathbf{CK}_{x_2}^{\mu_2}\big[x_1^{k}\big]\Big]\chi_{\pm}
 \\
 &=-\frac{(-1)^{k}}{2}\mathbf{CK}_{x_3}^{\mu_3}\Big[\widetilde{\underline{x}}^{N-k}\;\Big(\widetilde{\underline{x}}\;\widetilde{\underline{D}}-\widetilde{\underline{D}}\;\widetilde{\underline{x}}+1\Big)\,\mathbf{CK}_{x_2}^{\mu_2}\big[x_1^{k}\big]\Big]\chi_{\pm}
 \\
 &=-\frac{(-1)^{k}}{2}\mathbf{CK}_{x_3}^{\mu_3}\Big[\widetilde{\underline{x}}^{N-k}\;\Big(-2(\widetilde{E}+\gamma_2)+1\Big)\,\mathbf{CK}_{x_2}^{\mu_2}\big[x_1^{k}\big]\Big]\chi_{\pm},
\end{align*}
where in the last step the commutation relations \eqref{Commu-2} were used. Using the properties
\begin{align*}
 \widetilde{\underline{D}}\;\mathbf{CK}_{x_2}^{\mu_2}\big[x_1^{k}\big]=0,\quad \widetilde{\mathbb{E}}\;\mathbf{CK}_{x_2}^{\mu_2}\big[x_1^{k}\big]=k \;\mathbf{CK}_{x_2}^{\mu_2}\big[x_1^{k}\big],\qquad R_1R_2 \;\mathbf{CK}_{x_2}^{\mu_2}\big[x_1^{k}\big]=(-1)^{k}\;\mathbf{CK}_{x_2}^{\mu_2}\big[x_1^{k}\big].
\end{align*}
one finds that
\begin{align}
\label{KK}
 K_3\,\psi_{k,\pm}^{(N)}=(-1)^{k}(k+\mu_1+\mu_2+1/2)\,\psi_{k,\pm}^{(N)}.
\end{align}
Upon combining \eqref{QQ} and \eqref{KK}, it is seen that the spinors \eqref{Basis} satisfy the defining properties of the  basis vectors $\ket{N,k}$ for the representations of the Bannai--Ito algebra constructed in section 4. The spinors $\psi_{k,\pm}^{(N)}$ however possess an extra label $\pm$ associated to the eigenvalues of the symmetry operator $Z_3=\sigma_3 R_3$. Indeed, it is directly verified from the explicit expression \eqref{q} and \eqref{m} that one has
\begin{align*}
 Z_3\,\psi_{k,\pm}^{(N)}=\pm (-1)^{N-k}\,\psi_{k,\pm}^{(N)}.
\end{align*}
It follows that each of the two independent sets of basis vectors 
\begin{align*}
 \{\psi_{k,+}^{(N)}\;\rvert \; k=0,1,\ldots,N\},\qquad \{\psi_{k,-}^{(N)}\;\rvert \; k=0,1,\ldots,N\},
\end{align*}
supports a unitary $(N+1)$-dimensional irreducible representation of the Bannai--Ito algebra as constructed in section 4. As a consequence, the space of Dunkl monogenics $\mathcal{M}_{N}(\mathbb{R}^3)$ of degree $N$ can be expressed as a direct sum of two such representations. Since $\mathrm{dim}\,\mathcal{M}_{N}(\mathbb{R}^3)=2\times(N+1)$, the dimensions of the spaces match. 
\subsection{Normalized wavefunctions}
The wavefunctions \eqref{psi} can  be presented in a normalized fashion. We define
\begin{align}
\label{Wavefunctions}
 \Psi_{k,\pm}^{(N)}(x_1,x_2,x_3)=\Theta_{N,k}(x_1,x_2,x_3)\;\Phi_{k}(x_1,x_2)\;\chi_{\pm},
\end{align}
with
\begin{multline}
\label{Phi}
 \Phi_{k}(x_1,x_2)=\sqrt{\frac{\beta!\;\Gamma(\beta+\mu_1+\mu_2+1)}{2\,\Gamma(\beta+\mu_1+1/2)\,\Gamma(\beta+\mu_2+1/2)}}\;(x_1^2+x_2^2)^{\beta}
 \\
 \times
 \begin{cases}
  P_{\beta}^{(\mu_2-1/2,\mu_1-1/2)}\left(\frac{x_1^2-x_2^2}{x_1^2+x_2^2}\right) \mathbb{1}
& \text{$k=2\beta$,}
\\[.2cm]
\qquad \qquad \qquad +\frac{\sigma_1 x_1\,\sigma_2 x_2}{x_1^2+x_2^2}\,P_{\beta-1}^{(\mu_2+1/2,\mu_1+1/2)}\left(\frac{x_1^2-x_2^2}{x_1^2+x_2^2}\right), &
\\[.3cm]
\sqrt{\frac{\beta+\mu_2+1/2}{\beta+\mu_1+1/2}}\;x_1\;P_{\beta}^{(\mu_2-1/2,\mu_1+1/2)}\left(\frac{x_1^2-x_2^2}{x_1^2+x_2^2}\right)\mathbb{1} & \text{$k=2\beta+1$},
\\[.2cm]
\qquad \qquad \qquad  -\sqrt{\frac{\beta+\mu_1+1/2}{\beta+\mu_2+1/2}}\;\sigma_2\,\sigma_1\,x_2\; P_{\beta}^{(\mu_2+1/2,\mu_1-1/2)}\left(\frac{x_1^2-x_2^2}{x_1^2+x_2^2}\right),&
 \end{cases}
\end{multline}
and where
\begin{multline}
\label{Theta}
 \Theta_{N,k}(x_1,x_2,x_3)=\sqrt{\frac{\beta!\;\Gamma(\beta+k+\mu_1+\mu_2+\mu_3+3/2)}{\Gamma(\beta+\mu_3+1/2)\;\Gamma(\beta+k+\mu_1+\mu_2+1)}}
 \\
 \times
 \begin{cases}
  P_{\beta}^{(\mu_3-1/2,k+\mu_1+\mu_2+1)}\left(\frac{x_1^2+x_2^2-x_3^2}{x_1^2+x_2^2+x_3^2}\right)\mathbb{1} & \text{$N-k=2\beta$,}
  \\[.2cm]
  \qquad \qquad +\frac{(\sigma_1x_1+\sigma_2x_2)\;\sigma_3 x_3}{x_1^2+x_2^2+x_3^2}\;P_{\beta-1}^{(\mu_3+1/2,k+\mu_1+\mu_2+1)}\left(\frac{x_1^2+x_2^2-x_3^2}{x_1^2+x_2^2+x_3^2}\right), & 
  \\[.4cm]
  \sqrt{\frac{\beta+\mu_3+1/2}{\beta+k+\mu_1+\mu_2+1}}\;(\sigma_1x_1+\sigma_2 x_2)\;P_{\beta}^{(\mu_3-1/2,k+\mu_1+\mu_2+1)}\left(\frac{x_1^2+x_2^2-x_3^2}{x_1^2+x_2^2+x_3^2}\right) & \text{$N-k=2\beta+1$.}
  \\[.2cm]
  \qquad \qquad \qquad -\sqrt{\frac{k+\beta+\mu_1+\mu_2+1}{\beta+\mu_3+1/2}}\;\sigma_3 x_3\;P_{\beta}^{(\mu_3+1/2,k+\mu_1+\mu_2)}\left(\frac{x_1^2+x_2^2-x_3^2}{x_1^2+x_2^2+x_3^2}\right),
 \end{cases}
\end{multline}
In \eqref{Phi} and \eqref{Theta}, the symbol $\mathbb{1}$ stands for the $2\times 2$ identity operator and $\Gamma(x)$ is the standard Gamma function \cite{2001_Andrews&Askey&Roy}. Introduce the scalar product
\begin{align}
\label{Scalar}
 \langle \Lambda, \Psi \rangle=\int_{S^2} (\Lambda^{\dagger}\cdot \Psi) \;h(x_1,x_2,x_3)\;\mathrm{d}x_1\mathrm{d}x_2\mathrm{d}x_3,
\end{align}
where $h(x_1,x_2,x_3)$ is the $\mathbb{Z}_2^3$ invariant weight function \cite{2014_Dunkl&Xu}
\begin{align*}
 h(x_1,x_2,x_3)=|x_1|^{2\mu_1}|x_2|^{2\mu_2}|x_3|^{2\mu_3}.
\end{align*}
It is directly verified (see for example \cite{2013_Genest&Ismail&Vinet&Zhedanov_JPhysA_46_145201}) that the spherical Dirac-Dunkl operator $\Gamma$ and its symmetry operators $K_i$, $Z_i$ are self-adjoint with respect to the scalar product \eqref{Scalar}. Upon writing the wavefunctions \eqref{Wavefunctions} in the spherical coordinates, it follows from the orthogonality relation of the Jacobi polynomials (see for example \cite{2010_Koekoek_&Lesky&Swarttouw}) that the wavefunctions \eqref{Wavefunctions} satisfy the orthogonality relation
\begin{align*}
 \langle \Psi_{k',j}^{(N')},\Psi_{k,j'}^{(N)}\rangle=\delta_{kk'}\delta_{NN'}\delta_{jj'}.
\end{align*}
\subsection{Role of the Bannai--Ito polynomials}
Let us briefly discuss the role played by the Bannai--Ito polynomials in the present picture. It is known that these polynomials arise as overlap coefficients between the respective eigenbases of any pair of generators of the Bannai--Ito algebra in the representations \eqref{Eigen-Setup}  \cite{2014_Genest&Vinet&Zhedanov_ProcAmMathSoc_142_1545, 2012_Tsujimoto&Vinet&Zhedanov_AdvMath_229_2123}. We introduce the basis $\Upsilon_{s,\pm}^{(N)}$ defined by
\begin{align}
\label{Upsilon}
 \Upsilon_{s,\pm}^{(N)}=\widetilde{\Theta}_{N,s}(x_2,x_3,x_1)\;\widetilde{\Phi}_{s}(x_2,x_3)\;\chi_{\pm},\qquad s=0,\ldots,N,
\end{align}
where $\widetilde{\Theta}$ and $\widetilde{\Phi}$ are obtained from \eqref{Phi} and \eqref{Theta} by applying the permutation $(\mu_1,\mu_2,\mu_3)\rightarrow (\mu_2,\mu_3,\mu_1)$. It is easily seen from \eqref{BI-Operators} and \eqref{Gamma-k} that the wavefunctions \eqref{Upsilon} satisfy the eigenvalue equations
\begin{align*}
 (\Gamma+1)\,\Upsilon_{s,\pm}^{(N)}&=(N+\mu_1+\mu_2+\mu_3+1)\,\Upsilon_{s,\pm}^{(N)},
 \\
 K_1\,\Upsilon_{s,\pm}^{(N)}&=(-1)^{s}(s+\mu_2+\mu_3+1/2)\,\Upsilon_{s,\pm}^{(N)},
 \\
 \sigma_3 R_3\,\Upsilon_{s,\pm}^{(N)}&=\pm (-1)^{N-s}\,\Upsilon_{s,\pm}^{(N)}.
\end{align*}
With the scalar product \eqref{Scalar}, the overlap coefficients between the bases $\Psi_{k,\pm}^{(N)}$ and $\Upsilon_{s,\pm}^{(N)}$ are defined as
\begin{align*}
\langle \Upsilon_{s,q}^{(N)},\, \Psi_{k,r}^{(N)} \rangle=W_{s,k;q,r}^{(N)}.
\end{align*}
The coefficients $W_{s,k;q,r}^{(N)}$ can be expressed in terms of the Bannai--Ito polynomials (see \cite{2015_Genest&Vinet&Zhedanov_CommMathPhys}).
\section{Conclusion}
In this paper, we considered the Dirac--Dunkl operator on the two-sphere associated to the $\mathbb{Z}_2^3$ Abelian reflection group. We have obtained its symmetries and shown that they generate the Bannai--Ito algebra. We have built the relevant representations of the Bannai--Ito algebra using ladder operators. Finally, using a Cauchy-Kovalevskaia extension theorem, we have constructed the eigenfunctions of the spherical Dirac--Dunkl operator and we have shown that they transform according to irreducible representations of the Bannai--Ito algebra.

As observed in this paper, the formulas \eqref{BI-Algebra} can be considered as a three-parameter deformation of the algebra $sl_{2}$ and as such, it can be considered to have rank one. It would of great interest in the future to generalize the Bannai--Ito algebra to arbitrary rank. In that regard, the study of the Dirac--Dunkl operator in $n$ dimensions associated to the $\mathbb{Z}_2^{n}$ reflection group is interesting.
\section*{Acknowledgements}
H.D.B. is grateful for the hospitality extended to him by the Centre de recherches mathématiques, where this research was carried. V.X.G. holds an Alexander--Graham--Bell fellowship from the Natural Science and Engineering Research Council of Canada (NSERC). The research of L.V. is supported in part by NSERC.
\begin{multicols}{2}
\footnotesize
%\bibliographystyle{plain}
%\bibliography{/home/vincent/Documents/References/Bibliography_VXG.bib}

\begin{thebibliography}{10}

\bibitem{2001_Andrews&Askey&Roy}
G.~Andrews, R.~Askey, and R.~Roy.
\newblock {\em {Special functions}}, volume~71 of {\em {Encyclopedia of
  Mathematics and its Applications}}.
\newblock Cambridge University Press, 2001.

\bibitem{2015_DeBie&Genest&Tsujimoto&Vinet&Zhedanov}
H.~De Bie, V.~X. Genest, S.~Tsujimoto, L.~Vinet, and A.~Zhedanov.
\newblock {The Bannai-Ito algebra and some applications}.
\newblock {\em Journal of Physics: Conference Series}, 2015.

\bibitem{2012_DeBie&Orsted&Somberg&Soucek_TransAmerMathSoc_364_3875}
H.~De Bie, B.~{\O}rsted, P.~Somberg, and V.~Sou\v{c}ek.
\newblock {Dunkl operators and a family of realizations of
  $\mathfrak{osp}(1|2)$}.
\newblock {\em Transactions of the American Mathematical Society},
  364(7):3875--3902, July 2012.

\bibitem{2011_DeBie&DeSchepper_BullBelMathSoc_18_193}
H.~De Bie and N.~De Schepper.
\newblock {Clifford-Gegenbauer polynomials related to the Dunkl Dirac
  operator}.
\newblock {\em The Bulletin of the Belgian Mathematical Society-Simon Stevin},
  18(2):193--214, May 2011.

\bibitem{1982_Brackx&Delanghe&Sommen}
F.~Brackx, R.~Delanghe, and F.~Sommen.
\newblock {\em {Clifford analysis}}.
\newblock Addison-Wesley, 1982.

\bibitem{2006_Cerejeiras&Kahler&Ren_ComplexVar&EllEqs_51_487}
P.~Cerejeiras, U.~K{\"a}hler, and G.~Ren.
\newblock {Clifford analysis for finite reflection groups}.
\newblock {\em Complex Variables and Elliptic Equations}, 51(5-6):487--495,
  October 2006.

\bibitem{2013_Dai&Xu}
F.~Dai and Y.~Xu.
\newblock {\em {Approximation Theory and Harmonic Analysis on Spheres and
  Balls}}.
\newblock Springer, 2013.

\bibitem{1988_Dunkl_MathZeit_197_33}
C.~Dunkl.
\newblock {Reflection groups and orthogonal polynomials on the sphere}.
\newblock {\em Mathematische Zeitschrift}, 197(13):33--60, 1988.

\bibitem{1989_Dunkl_TransAmerMathSoc_311_167}
C.~Dunkl.
\newblock {Differential-difference operators associated to reflection groups}.
\newblock {\em Transactions of the American Mathematical Society},
  311(13):167--183, January 1989.

\bibitem{1991_Dunkl_CanJMath_43_1213}
C.~Dunkl.
\newblock {Integral kernels with reflection group invariance}.
\newblock {\em Canadian Journal of Mathematics}, 43(6):1213--1227, December
  1991.

\bibitem{2014_Dunkl&Xu}
C.~Dunkl and Y.~Xu.
\newblock {\em {Orthogonal Polynomials of Several Variables}}.
\newblock Cambridge University Press, 2\textsuperscript{nd} edition, 2014.

\bibitem{2000_Frappat&Sciarrino&Sorba}
L.~Frappat, A.~Sciarrino, and P.~Sorba.
\newblock {\em {Dictionary on Lie Algebras and Superalgebras}}.
\newblock Academic Press, 2000.

\bibitem{2013_Genest&Ismail&Vinet&Zhedanov_JPhysA_46_145201}
V.~X. Genest, M.~Ismail, L.~Vinet, and A.~Zhedanov.
\newblock {The Dunkl oscillator in the plane: I. Superintegrability, separated
  wavefunctions and overlap coefficients }.
\newblock {\em Journal of Physics A: Mathematical and Theoretical},
  46(14):145201, April 2013.

\bibitem{2014_Genest&Ismail&Vinet&Zhedanov_CommMathPhys_329_999}
V.~X. Genest, M.~Ismail, L.~Vinet, and A.~Zhedanov.
\newblock {The Dunkl Oscillator in the Plane II: Representations of the
  Symmetry Algebra}.
\newblock {\em Communications in Mathematical Physics}, 329(3):999--1029,
  August 2014.

\bibitem{2013_Genest&Vinet&Zhedanov_SIGMA_9_18}
V.~X. Genest, L.~Vinet, and A.~Zhedanov.
\newblock {Bispectrality of the Complementary Bannai--Ito polynomials}.
\newblock {\em SIGMA}, 9:18, March 2013.

\bibitem{2014_Genest&Vinet&Zhedanov_JPhysA_47_205202}
V.~X. Genest, L.~Vinet, and A.~Zhedanov.
\newblock {The Bannai-Ito algebra and a superintegrable system with reflections
  on the two-sphere}.
\newblock {\em Journal of Physics A: Mathematical and Theoretical},
  47(20):205202, May 2014.

\bibitem{2014_Genest&Vinet&Zhedanov_ProcAmMathSoc_142_1545}
V.~X. Genest, L.~Vinet, and A.~Zhedanov.
\newblock {The Bannai-Ito polynomials as Racah coefficients of the $sl_{-1}(2)$
  algebra}.
\newblock {\em Proceedings of the American Mathematical Society},
  142(5):1545--1560, May 2014.

\bibitem{2014_Genest&Vinet&Zhedanov_JPhysConfSer_512_012010}
V.~X. Genest, L.~Vinet, and A.~Zhedanov.
\newblock {The Dunkl oscillator in three dimensions}.
\newblock {\em Journal of Physics: Conference Series}, 512(13):012010, 2014.

\bibitem{2015_Genest&Vinet&Zhedanov_CommMathPhys}
V.~X. Genest, L.~Vinet, and A.~Zhedanov.
\newblock {A Laplace-Dunkl equation on $S^2$ and the Bannai--Ito algebra}.
\newblock {\em Communications in Mathematical Physics (in Press)}, 2015.

\bibitem{2008_Graczyk&Rosler&Yor}
P.~Graczyk, M.~R{\"o}sler, and M.~Yor, editors.
\newblock {\em {Harmonic and Stochastic Analysis of Dunkl Processes}}.
\newblock Hermann, 2008.

\bibitem{2010_Koekoek_&Lesky&Swarttouw}
R.~Koekoek, P.~A. Lesky, and R.~F. Swarttouw.
\newblock {\em {Hypergeometric orthogonal polynomials and their
  $q$-analogues}}.
\newblock Springer, 2010.

\bibitem{1995_Lesniewski_JMathPhys_36_1457}
A.~Lesniewski.
\newblock {A remark on the Casimir elements of Lie superalgebras and quantized
  Lie superalgebras}.
\newblock {\em Journal of Mathematical Physics}, 36(3):1457--1461, March 1995.

\bibitem{2009_Orsted&Somberg&Soucek_AdvApplCliffAlg_19_403}
B.~{\O}rsted, P.~Somberg, and V.~Sou\v{c}ek.
\newblock {The Howe duality for the Dunkl version of the Dirac operator}.
\newblock {\em Advances in Applied Clifford Algebras}, 19(2):403--415, July
  2009.

\bibitem{2003_Rosler}
M.~R{\"o}sler.
\newblock {Dunkl Operators: Theory and Applications}.
\newblock In {\em {Orthogonal Polynomials and Special Functions}}, volume 1817
  of {\em {Lecture Notes in Mathematics }}, pages 93--135. Springer, 2003.

\bibitem{1998_Rosler&Voit_AdvApplMath_21_575}
M.~R{\"o}sler and M.~Voit.
\newblock {Markov Processes Related with Dunkl Operators}.
\newblock {\em Advances in Applied Mathematics}, 21(4), November 1998.

\bibitem{2011_Tsujimoto&Vinet&Zhedanov_SIGMA_7_93}
S.~Tsujimoto, L.~Vinet, and A.~Zhedanov.
\newblock {From $sl_{q}(2)$ to a parabosonic Hopf algebra}.
\newblock {\em SIGMA}, 7:93--105, 2011.

\bibitem{2012_Tsujimoto&Vinet&Zhedanov_AdvMath_229_2123}
S.~Tsujimoto, L.~Vinet, and A.~Zhedanov.
\newblock {Dunkl shift operators and Bannai--Ito polynomials}.
\newblock {\em Advances in Mathematics}, 229(4):2123--2158, March 2012.

\bibitem{2000_VanDiejen&Vinet}
L.~Vinet and J.~F.~Van Diejen, editors.
\newblock {\em {Calogero-Moser-Sutherland models}}.
\newblock Springer, 2000.

\end{thebibliography}

\end{multicols}
\end{document}